\documentclass[12pt,a4paper]{article}
\usepackage{amsmath}
\usepackage{amsfonts}
\usepackage{amssymb}
\usepackage{par skip}
\usepackage{graphicx}
\usepackage[english]{babel}
\usepackage{amsthm}
\usepackage{adjustbox}
\newcounter{thm}

\usepackage[colorlinks,bookmarksopen,bookmarksnumbered,citecolor=red,urlcolor=red]{hyperref} %

\newtheorem{Theorem}[thm]{Theorem}
\newtheorem{Corollary}{Corollary}[thm]
\newcommand*{\permcomb}[4][0mu]{{{}^{#3}\mkern#1#2_{#4}}}

\newcommand*{\comb}[1][-1mu]{\permcomb[#1]{C}}

\title{ON SIZE BIASED KUMARASWAMY DISTRIBUTION}
\author{Dreamlee Sharma $^{1,}$* and Tapan Kumar Chakrabarty $^{1}$ \\ $^{1}$\footnotesize{\textit{Department of Statistics, North-Eastern Hill University, Shillong - 793022, Meghalaya, India}} \\ \small{Contact: *dreamleesharma@yahoo.in, tapankumarchakrabarty@gmail.com}}

\begin{document}
\maketitle

\abstract{In this paper, we introduce and study the size-biased form of Kumaraswamy distribution. The Kumaraswamy distribution  which has drawn considerable attention in hydrology and related areas was proposed by Kumarswamy \cite{Kum80}. The new distribution is derived under size-biased probability of sampling taking the weights as the variate values. Various distributional and characterizing properties of the model are studied. The methods of maximum likelihood and matching quantiles estimation are employed to estimate the parameters of the proposed model. Finally, we apply the proposed model to simulated and real data sets.}


\section{Introduction}

The concept of weighted distribution was first introduced by Fisher \cite{Fis34} to model ascertainment bias, and was later formalized in a unifying theory by Rao \cite{65Rao}. Let $X$ be a random variable of interest such that $X \sim f(x;\theta)$, where $\theta$ is a vector of parameters. Under equal probability sampling, the estimation of the parameter $\theta$ can be made with an abundance of methods. However, under size-biased schemes, the probability of sampling an individual is proportional to $X^r$ provided that $ E_{\theta}(X^r)<\infty $ for all $\theta$. In situations like this, the weighted probability density function is defined as
\begin{equation}\label{wdist}
f_r(x,\theta)= \frac{x^r f(x,\theta)}{\mu_r^{'}}
\end{equation}
where \\
$${\mu_r}^{'}=\int {x^r f(x;\theta)dx}$$ \\
in place of $f(x;\theta)$ can be used. The weighted distributions have varieties of uses in various fields. A number of papers have appeared implicitly using the concepts of weighted and size-biased sampling distributions. Patil and Rao \cite{77PatR} have briefly surveyed the applications of weighted and size-biased distributions. Size-biased distributions arise naturally in a range of sampling and modeling problems in forestry \cite{03Gov}. They also occur in applications spanning domains including environmental sciences, econometrics, human demography and biomedical sciences \cite{ PatR78, She72}. To have an idea of their applications, one can refer to, \cite{DenP84, DucG15,  LapB87, Maget99, Taiet95, Van86, 75War, Yuet12}.
\\
\medskip 
When the probability of observing a positive-valued random variable is proportional to the value of the variable the resultant is size-biased distribution. Size-biased distributions of order $1$ is a special case of the weighted distribution defined in (\ref{wdist}) with weight as $x$. In this paper, the term size-biased distribution will be used to indicate the size-biased distribution of order $1$. Thus taking $r = 1$, in (\ref{wdist}) we obtain the size biased distribution which is given by the \textit{p.d.f.}
\begin{equation}\label{sbdist}
g(x,\theta)=\frac{xf(x,\theta)}{\mu_1^{'}}
\end{equation}

\section{The Size Biased Kumaraswamy Distribution}

The Kumaraswamy distribution \cite{Kum80} is similar to the Beta distribution, but much simpler to use especially in simulation studies due to the simple closed form of both its probability density function and cumulative distribution function. This distribution is mainly used for variables that are lower and upper bounded. The probability density function (pdf) of the Kumaraswamy distribution (Kum) is given by
\begin{equation}\label{pkum}
\begin{aligned}
g(x;a,b) &=abx^{a-1} (1-x^a )^{b-1},\;\;\; \textit{for} \;\; o<x<1\\                                    
  &= 0,\;\; \textit{otherwise}
\end{aligned}	
\end{equation}			       
\medskip 
where, $ a > 0$ and $b > 0$ are the two shape parameters.\\
The $r^{th}$ order raw moment of the Kum is given by
\begin{equation*}
{\mu_r}^{'}= bB(1+\frac{r}{a},b)           
\end{equation*}
where $B(1+\frac{r}{a},b)$ is a beta function defined by the integral
\begin{align*}
B(\alpha ,\beta)=\int_{0}^{1}x^{\alpha -1}(1-x)^{\beta - 1}dx         
\end{align*}
Thus, the expectation of the Kum is given by
\begin{equation}\label{mukum}
{\mu_1}^{'}= bB(1+\frac{1}{a},b)=\mu\;(say)                                                
\end{equation}
Thus using the relation (\ref{sbdist}) and (\ref{mukum}), the pdf of the SBKD is obtained as:
\begin{equation}\label{pSBKD}
f(x;a,b) = \frac{ax^a(1-x^a)^{b-1}}{B(1+\frac{1}{a},b)}\; ;\;   0<x<1
\end{equation}
Ducey and Gove \cite{DucG15} have obtained the weighted distribution of the Generalized Beta I (GBI), the Generalized Beta II (GBII) and the Generalized Gamma (GG) distributions and have shown that the GBI, the GBII, the GG distributions are form invariant under size biased scheme. The Kumaraswamy distribution is a distribution in the GBI($\alpha, \beta, p, q$) family of distributions \cite{McD84}. So that the SBKD is also a special case of the GBI distribution for $\alpha = a > 0, \beta = 1, p = 1+ \frac{1}{a}, q = b$.
\subsection{Special Cases}
\begin{itemize}
\item[\textit{1.}]	Taking $a=1$ in (\ref{pSBKD}) we get,
		$$f(x;b)= \frac{1}{B(2,b)} x^{2-1} (1-x)^{b-1}  ;  \ \ \ 0<x<1$$
		Thus the SBKD reduces to a Beta-I distribution with parameters $2$ and $b$.
\item[\textit{2.}]	Taking $b=1$ in (\ref{pSBKD}) we get
		$$f(x;a)=(a+1) x^a  ;  \ \ \ 0<x<1$$

\item[\textit{3.}]	Taking $a=1$ and $b=1$ in (\ref{pSBKD}), the SBKD reduces to a special case of the Triangular distribution
		$$f(x)= 2x;  \ \ \ 0 < x < 1 $$
\end{itemize}

\subsection{Shape of the distribution}
The SBKD is a Beta$(2, b)$ distribution for $a=1$. Hence, for any $b$ and a fixed $a=1$, the distributional shape of SBKD will be like that of a Beta$(2, b)$ distribution. therefore for $a=1$, the following shapes will be obtained.

\begin{itemize}
\item[\textit{1.}]	We know, a Beta I distribution is always symmetric if both the parameters are equal. Hence for $a=1$, the SBKD is symmetric if $b=2$. 
\item[\textit{2.}]	We know the Beta$(2,1)$ distribution is the Right-Triangular distribution with right angle at the right end, at $x = 1$ and is a straight line with slope $+2$. Hence the SBKD$(1,1)$ is also a Right Triangular distribution.
\item[\textit{3.}]   For a Beta$(a\geq 1, b<1)$, the Beta distribution is negatively skewed J-shaped curve. Hence the SBKD$(1, b < 1)$ is a J-shaped negatively skewed curve.
\item[\textit{4.}]   The Beta$(2, b)$ is unimodal and positively skewed for $b>2$ and negatively skewed for $1<b<2$ and hence the SBKD$(1, b)$ is also positively skewed for $b>2$ and negatively skewed for $1<b<2$. 
\end{itemize}
Figure \ref{bet} gives a plot of the possible shapes of the distribution for $a=1$. \\

\begin{figure}[ht]
	\centering
	\includegraphics[width=12cm]{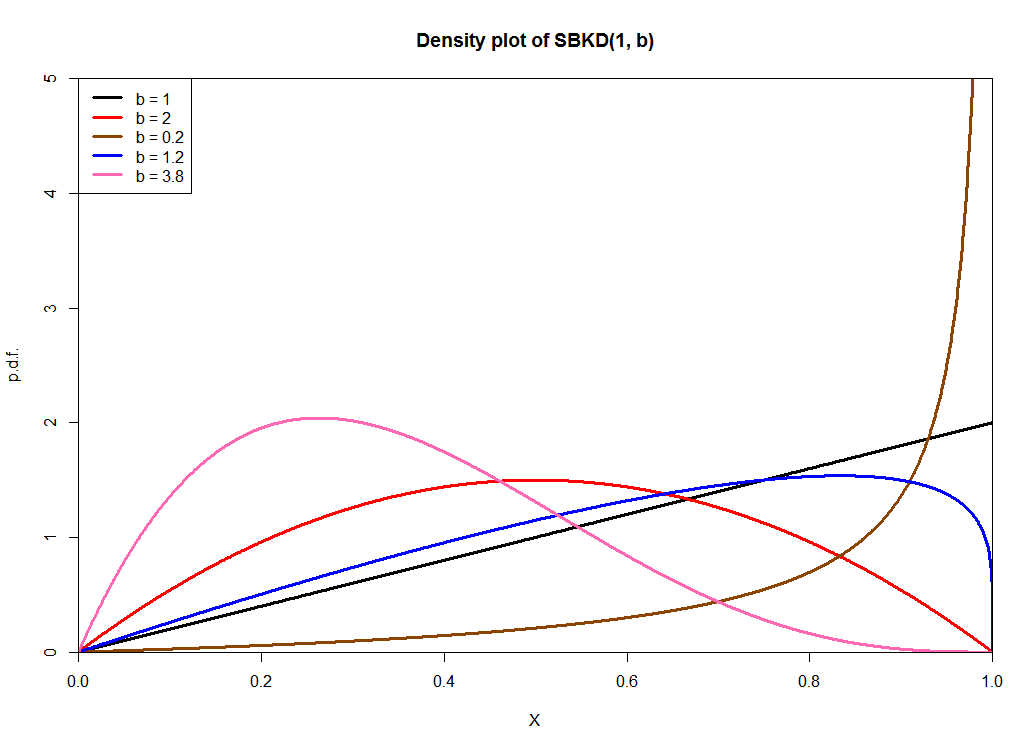}
	\caption{Density plot of SBKD$(1, b)$ for various values of $b$}
	\label{bet}
\end{figure}  

Some of the possible density plots of the SBKD for $a<1$ and $a>1$ is given respectively in Figure \ref{dena} and \ref{dena1}.

\begin{figure}[ht]
\centering
\includegraphics[width=12cm]{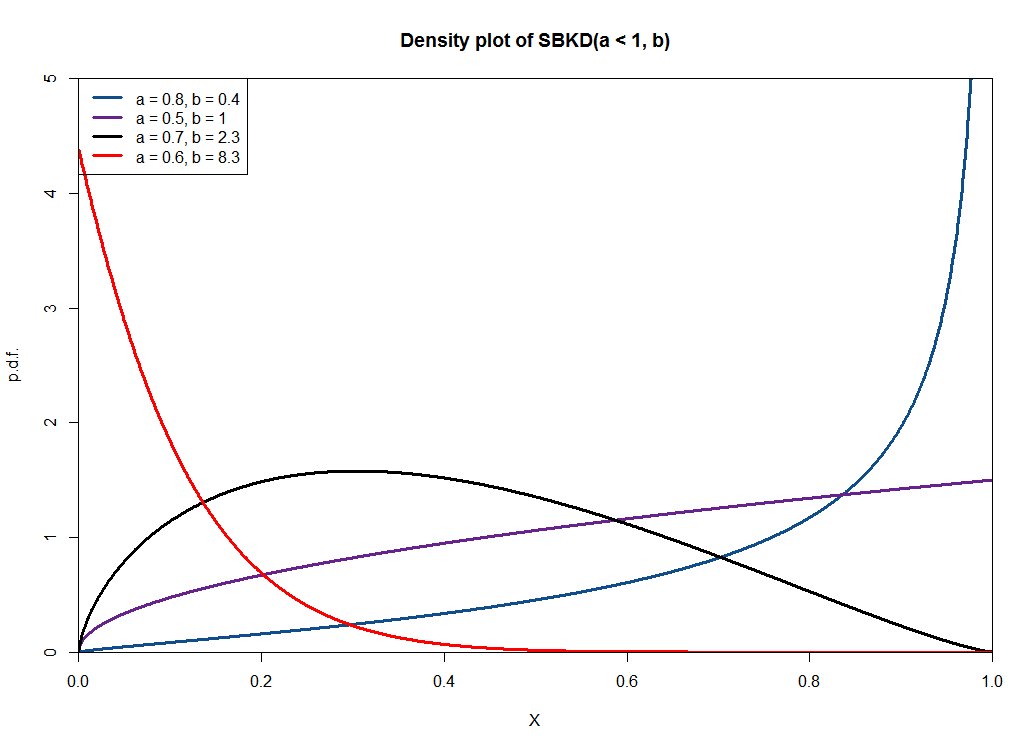}
\caption{Density plot of SBKD for $a<1$}
\label{dena}
\end{figure}   
\begin{figure}[ht]
\centering
\includegraphics[width=12cm]{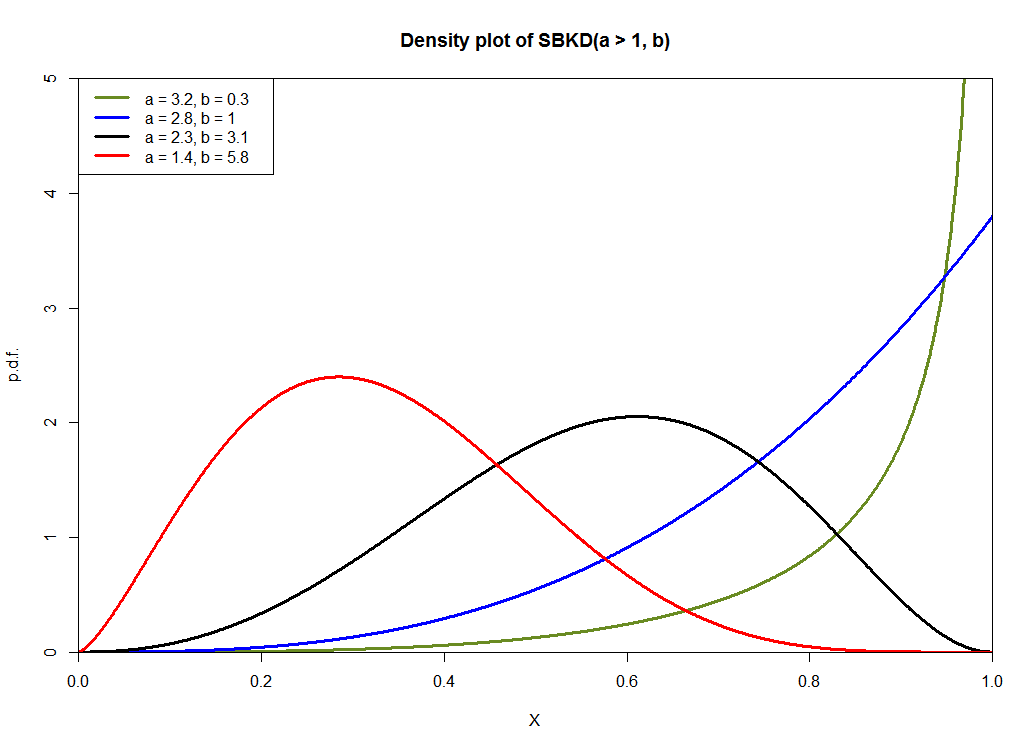}
\caption{Density plot of SBKD for $a>1$}
\label{dena1}
\end{figure}

The possible shapes of the SBKD for $a<1$ and $a>1$ is discussed below:
\begin{itemize}
\item[]	For $a<1, b<1,$ the SBKD has a J-shaped negatively skewed density.
\item[]	For $a<1,b=1,$ the SBKD has an increasing density.
\item[]	For $a<1,b>1,$ the SBKD has either a unimodal positively skewed density or a reverse J-shaped positively skewed decreasing density.
\item[]   For $a>1,b<1,$ the SBKD has a J-shaped negatively skewed density.
\item[] 	For $a>1,b=1,$ the SBKD has a negatively skewed increasing density.
\item[] 	For $a>1,b>1$ the SBKD has a unimodal skewed density.
\end{itemize} 
\section{Properties of the Size-Biased Kumaraswamy Distribution}

\subsection{Cumulative distribution function of SBKD}
\begin{Theorem}
	Let $X \sim SBKD(a,b)$, then its cumulative distribution function (\textit{c.d.f.})is given by (\ref{csbkd})
		\begin{equation}\label{csbkd}
			F(x)=I_{x^a}\left(1+\frac{1}{a}, b \right)                                                    
		\end{equation}
	where, 
	$$I_{x^a}\left(1+\frac{1}{a}, b \right)= \frac{B\left(x^a;1+\frac{1}{a},b\right)}{B\left(1+\frac{1}{a},b\right)}$$
	 is the regularized incomplete beta function and is defined as the ratio of an incomplete beta function, $B(z;\alpha,\beta)=\int _0 ^ z x^{\alpha - 1}(1-x)^{\beta -1} dx$ and the complete beta function, $B(\alpha, \beta)$.
\end{Theorem}
\begin{proof}
	As $X \sim SBKD(a,b)$ so its \textit{p.d.f.} is given by (\ref{pSBKD}). Let $F(x)$ be the \textit{c.d.f.} of SBKD then by definition,
                                                $$F(x)= \int_0^x f(y)dy; \ \ \ 0<x<1      $$
	Thus substituting $f(y)$ from (\ref{pSBKD}) we have
	\begin{align*}
	F(x) &= \frac{1}{B \left(1+\frac{1}{a},b\right)}\int_0^x ay^a(1-y^a)^{b-1} dy  \\
	  &= I_{x^a} \left(1+\frac{1}{a},b\right)     
     \end{align*}
\end{proof}
\subsection{Quantile function of SBKD} \label{rngen}
	\begin{Theorem}
	Let $X \sim SBKD(a,b)$, then its quantile function is given by (\ref{qsbkd})
		\begin{equation}\label{qsbkd}
			Q(p)= \left[{I_p}^{-1}\left(1+\frac{1}{a}, b \right)\right]^\frac{1}{a}                                                   
		\end{equation}
		where, ${I_p}^{-1}(\alpha, \beta)$ is the inverse regularized beta function defined as 
		${I_p}^{-1}(\alpha, \beta)=w$ such that $I_w(\alpha, \beta)=p$ 
	\end{Theorem}
	\begin{proof}
		Let $F(x)=p$ be a \textit{c.d.f.} then the corresponding quantile function, $Q(p)$ is defined as
		\begin{equation} \label{qdef}
			Q(p)=F^{-1} (p) 
		\end{equation}
Therefore by using the relation (\ref{csbkd}) and (\ref{qdef}) the quantile function of the SBKD is 
	$$Q(p)= \left[{I_p}^{-1}\left(1+\frac{1}{a}, b \right)\right]^\frac{1}{a}$$
	\end{proof}
	
	\begin{Corollary}
	The median of the SBKD is
	$$Q(0.5)= \left[{I_{0.5}}^{-1}\left(1+\frac{1}{a}, b \right)\right]^\frac{1}{a}$$
	\end{Corollary}
\subsubsection{Random number generation} 
	Using the quantile function of the SBKD as defined in (\ref{qsbkd}), a random sample of size $n$ can be simulated. Let $U$ be a uniform ($U(0,1)$) \textit{r.v.} and let $Q(p)$, $0\leq p \leq 1$ be the quantile function of SBKD, then by uniform transformation rule, \cite{00Gil} the variable $X$, where $x = Q(u)$, has a distribution with quantile function $Q(p)$. Thus, by using the uniform transformation rule, a random sample of size $n$ can be easily simulated from the SBKD by generating a random sample of the same size from a $U(0,1)$ distribution.
\subsection{Moment generating function of SBKD}
	\begin{Theorem}
		Let $X \sim SBKD(a,b)$ then the moment generating function, $M_X (t)$ of $X$ is given by
		\begin{equation} \label{mgfsbkd}
		M_X (t)= \sum_{i=0}^\infty \frac{t^i}{i!}\frac{B\left(1+\frac{i+1}{a},b\right)}{B\left(1+\frac{1}{a},b\right)}
		\end{equation}
	\end{Theorem}
	\begin{proof}
		By definition, the moment generating function \textit{m.g.f.} of a \textit{r.v.} $X$ is given by
		$$M_X (t)=E(e^{tx} )= \int_{-\infty}^{\infty} e^{tx} f(x)dx$$
		Thus, for a SBKD, the \textit{m.g.f.} is
		
	\begingroup\makeatletter\def\f@size{8}\check@mathfonts
\def\maketag@@@#1{\hbox{\m@th\normalsize\normalfont#1}}
		\begin{align*}
			M_X (t) &= \frac{1}{B\left(1+\frac{1}{a}, b \right)} \int_{0}^{1} e^{tx} ax^a (1-x^a)^{b-1} dx \\
			& = \frac{1}{B\left(1+\frac{1}{a}, b \right)} \int_{0}^{1} \left(1+tx+\frac{t^2 x^2}{2!}+ \frac{t^3 x^3}{3!}+ ... + \frac{t^n x^n}{n!} + ... \right)ax^a (1-x^a)^{b-1} dx \\
			& = \sum_{i=0}^\infty \frac{t^i}{i!}\frac{B\left(1+\frac{i+1}{a},b\right)}{B\left(1+\frac{1}{a},b\right)}
		\end{align*} \endgroup
	\end{proof}
	
	\begin{Corollary}
		The cumulant generating function, $K_X(t)$ of the SBKD is given by
$$ K_X (t)= ln \sum_{i=0}^\infty \frac{t^i}{i!}\frac{B\left(1+\frac{i+1}{a},b\right)}{B\left(1+\frac{1}{a},b\right)} $$
	\end{Corollary}
	
	\begin{Corollary} \label{rawmom}
	The $r^{th}$ order raw moment of SBKD is
	\begin{equation}
		{\mu_r}^{'} = \frac{B\left( 1+\frac{r+1}{a},b\right)}{B\left( 1+\frac{1}{a},b\right)}
	\end{equation}
	\end{Corollary}
	
	\begin{Corollary}
		The mean \textit{i.e.} the $1^{st}$ order raw moment of SBKD is
	\begin{equation}\label{mnsbkd}
		{\mu_1}^{'} = \frac{B\left( 1+\frac{2}{a},b\right)}{B\left( 1+\frac{1}{a},b\right)}
	\end{equation}
	\end{Corollary}
\subsection{Moments of SBKD}

	\begin{Theorem}
	Let $X \sim SBKD(a,b)$, then the $r^{th}$ order raw moment ${\mu_r}^{'}$  and central moment $\mu_r$ are defined respectively by (\ref{rmom}) and (\ref{cenmom})
	\begin{equation} \label{rmom}
	{\mu_r}^{'} = \frac{B\left( 1+\frac{r+1}{a},b\right)}{B\left( 1+\frac{1}{a},b\right)}
	\end{equation}
	
	\begingroup\makeatletter\def\f@size{8}\check@mathfonts
\def\maketag@@@#1{\hbox{\m@th\normalsize\normalfont#1}}%
	\begin{align} \label{cenmom}
	 \mu_r = & \frac{1}{B\left( 1+\frac{1}{a},b\right)} \left[ B\left( 1+\frac{r+1}{a},b\right)-r\mu B\left( 1+\frac{r}{a},b\right)+\frac{r(r-1)}{2}\mu^2 B\left(1+\frac{r-1}{a},b\right)-... \right. \nonumber \\
	& \left. +(-1)^k\ \comb{r}{k}\mu ^k B\left( 1+\frac{r-k+1}{a},b\right)+...+(-1)^r \mu ^r B\left( 1+\frac{1}{a},b\right)\right]
	\end{align} \endgroup
	where, $\mu$ is the mean of the SBKD and is given by (\ref{mnsbkd}).
	\end{Theorem}
	
	\begin{proof}
	The proof for (\ref{rmom}) follows directly from the Corollary \ref{rawmom}.
	Now, the $r^{th}$ order central moment is defined as
	\begin{equation} \label{cnmom}
	\mu _r = E[(X-E(X))^r] = \int _{-\infty}^\infty (x-E(X))^r f(x) dx
	\end{equation}
	The $E(X)$ or the first order raw moment of the SBKD is given by (\ref{mnsbkd}). Let this be denoted by $\mu $. Thus using the relation (\ref{pSBKD}), (\ref{mnsbkd}) and (\ref{cnmom}), the $r^{th}$ order central moment of the SBKD is obtained as
	\begingroup\makeatletter\def\f@size{8}\check@mathfonts
\def\maketag@@@#1{\hbox{\m@th\normalsize\normalfont#1}}%
	\begin{align*}
	\mu _r & = \frac{1}{B\left(1+\frac{1}{a}, b \right)} \int _{-\infty}^{\infty} (x-\mu )^r ax^a (1-x^a)^{b-1} dx \\
		   & = \frac{1}{B\left(1+\frac{1}{a}, b \right)} \int _{0}^{1} \left[(x^r - \comb{r}{1} x^{r-1}\mu + ...+(-1)^k \  \comb{r}{k} x^{r-k} \mu ^k +...+(-1)^r \mu ^r \right] ax^a (1-x^a)^{b-1} dx \\
		   & = \frac{1}{B\left( 1+\frac{1}{a},b\right)} \left[ B\left( 1+\frac{r+1}{a},b\right)-r\mu B\left( 1+\frac{r}{a},b\right)+\frac{r(r-1)}{2}\mu^2 B\left(1+\frac{r-1}{a},b\right)-... \right. \nonumber \\
	& \ \ \ \ \left. +(-1)^k\ \comb{r}{k}\mu ^k B\left( 1+\frac{r-k+1}{a},b\right)+...+(-1)^r \mu ^r B\left( 1+\frac{1}{a},b\right)\right]
	\end{align*} \endgroup
	\end{proof}
	\begin{Corollary}
	The first four central moments are
	
	\begingroup\makeatletter\def\f@size{8}\check@mathfonts
\def\maketag@@@#1{\hbox{\m@th\normalsize\normalfont#1}}%
	\begin{flalign*}
	\mu _1 & = 0 \\
	\mu _2 & = \frac{B\left( 1+\frac{3}{a},b\right)}{B\left( 1+\frac{1}{a},b\right)}-\left[\frac{B\left( 1+\frac{2}{a},b\right)}{B\left( 1+\frac{1}{a},b\right)} \right]^2 \\
	\mu _3 & = \frac{B\left( 1+\frac{4}{a},b\right)}{B\left( 1+\frac{1}{a},b\right)}-3\frac{B\left( 1+\frac{2}{a},b\right)B\left( 1+\frac{3}{a},b\right)}{B\left( 1+\frac{1}{a},b\right)^2}+2\left[\frac{B\left( 1+\frac{2}{a},b\right)}{B\left( 1+\frac{1}{a},b\right)}\right]^3 \\
	\mu _4 & = \frac{B\left( 1+\frac{5}{a},b\right)}{B\left( 1+\frac{1}{a},b\right)}-4\frac{B\left( 1+\frac{2}{a},b\right)B\left( 1+\frac{4}{a},b\right)}{B\left( 1+\frac{1}{a},b\right)^2}+6\frac{B\left( 1+\frac{2}{a},b\right)^2B\left( 1+\frac{3}{a},b\right)}{B\left( 1+\frac{1}{a},b\right)^3}-3\left[\frac{B\left( 1+\frac{2}{a},b\right)}{B\left( 1+\frac{1}{a},b\right)}\right]^4
	\end{flalign*} \endgroup
	\end{Corollary}
\subsection{Skewness and kurtosis of SBKD}
 The skewness  of the SBKD is given by
 \begin{align*}
 S_k & = 3\frac{Mean - Median}{{\mu_2}^\frac{1}{2}} \\
 & = 3 \frac{B\left(1+\frac{2}{a},b\right)-\left[{I_{0.5}}^{-1}\left(1+\frac{1}{a}, b \right)\right]^\frac{1}{a}B\left(1+\frac{1}{a},b\right)}{ \left[ B\left(1+\frac{3}{a},b\right)B\left(1+\frac{1}{a},b\right)-B\left( 1+\frac{2}{a},b \right)^2\right]^\frac{1}{2}}
 \end{align*}
	The kurtosis of the SBKD is given by
	\begin{align*}
 \beta_2 & = \frac{\mu_4}{{\mu_2}^2} \\
 & = \frac{B\left(1+\frac{5}{a},b\right)-4\mu B\left(1+\frac{4}{a},b\right)+6\mu ^2 B\left(1+\frac{3}{a},b\right)-3\mu ^4 B\left(1+\frac{1}{a},b\right)}{\frac{B\left(1+\frac{3}{a},b\right)^2}{B\left(1+\frac{1}{a},b\right)}-2\mu ^2 B\left(1+\frac{3}{a},b\right)+\mu ^4 B\left(1+\frac{1}{a},b\right)}
 \end{align*}
 where, 
 $$\mu = \frac{B\left(1+\frac{2}{a},b\right)}{B\left(1+\frac{1}{a},b\right)}$$
 \subsection{Harmonic mean of SBKD}
	\begin{Theorem}
	Let $X \sim SBKD(a, b)$, then the harmonic mean of X is given by
		$$ H.M. = b \ B\left( 1+\frac{1}{a},b\right) $$
	\end{Theorem}
	\begin{proof}
	The harmonic mean of a \textit{r.v} $X$ is given as
	$$\frac{1}{H.M.} = \int _{-\infty}^{\infty} \frac{1}{x} f(x) dx$$
	Thus for a SBKD, the $H.M.$ is
	\begin{align*}
	\frac{1}{H.M.} & = \frac{1}{B\left( 1+\frac{1}{a},b\right)}\int _{0}^{1} \frac{1}{x} ax^a(1-x^a)^{b-1}  dx \\
	& = \frac{1}{b \ B\left( 1+\frac{1}{a},b\right)} \\
	or \ \ H.M. & = b \ B\left( 1+\frac{1}{a},b\right)
	\end{align*}
	\end{proof}		
	
\subsection{The survival and hazard function}
	The survival function of a SBKD is given by
	\begin{equation}
		S(t) = 1- F(t) = 1-I_{t^a}\left(1+\frac{1}{a}, b \right)
	\end{equation}
	
The hazard function of the SBKD is given by
$$h(t)=\frac{f(t)}{S(t)}=\frac{at^a(1-t^a)^{b-1}}{B\left( 1+\frac{1}{a},b\right)-B\left( t^a;1+\frac{1}{a},b\right)}$$

\section{Parameter estimation of SBKD}\label{Esti}
\subsection{Method of maximum likelihood estimation}

The method of maximum likelihood estimation (MLE) selects the set of values of the model parameters that maximizes the likelihood function. By definition of the method of maximum likelihood estimation, it is required to first specify the joint density function for all observations. For a random sample of size $n$ from SBKD, the likelihood function is given by
$$L = \prod _{i=1}^{n} \frac{a{x_i}^a(1-{x_i}^a)^{b-1}}{B\left( 1+\frac{1}{a},b\right)} $$
or equivalently,
\begin{align}\label{liksbkd}
ln(L) & = \sum _{i=1}^{n} ln \left(\frac{a{x_i}^a(1-{x_i}^a)^{b-1}}{B\left( 1+\frac{1}{a},b\right)}\right) \nonumber \\
	& = n \ ln(a) + a \sum _i ln(x_i) + (b-1)\sum _i ln(1-{x_i}^a) - n \ ln\ B\left( 1+\frac{1}{a},b\right) 
\end{align}
To obtain the MLE of the SBKD, (\ref{liksbkd}) is differentiated w.r.t. $a$ and $b$ and then equated to $0$. Hence the likelihood equations are

\begingroup\makeatletter\def\f@size{8}\check@mathfonts
\def\maketag@@@#1{\hbox{\m@th\normalsize\normalfont#1}}
\begin{equation}\label{mlesbkd}
\begin{aligned}
& \frac{n}{\widehat{a}}+\sum _i ln(x_i)-\widehat{a}(\widehat{b}-1)\sum _i \frac{{x_i}^{\widehat{a}-1}}{1-{x_i}^{\widehat{a}}} + \frac{n}{{\widehat{a}}^2} \left[ \psi \left(1+\frac{1}{\widehat{a}} \right)-\psi \left(1+\frac{1}{\widehat{a}}+\widehat{b} \right) \right] = 0 \\
& \sum _i ln(1-{x_i}^{\widehat{a}}) - n \left[ \psi (b)-\psi \left(1+\frac{1}{\widehat{a}}+\widehat{b} \right) \right] = 0
\end{aligned}
\end{equation} \endgroup
where, $\psi (.)$ is the digamma function given by the logarithmic derivative of the gamma function. \\
The set of equations (\ref{mlesbkd}) can be solved by using numerical methods.

\subsection{Method of quantile matching estimation}
The method of matching quantiles, an iterative procedure based on the ordinary least squares estimation (OLS) computes matching quantile estimation (MQE). The method of matching quantiles is based on matching theoretical quantiles of the parametric distribution against the empirical quantiles for specified probabilities, \cite{09Tse}. The basic idea is to match the distribution of total counterpart portfolio by that of a selected portfolio. We choose the representative portfolio to minimize the mean squared difference between the quantiles of the two distributions across all levels. This leads to the matching quantiles estimation (MQE). If $\tilde{Q_p}$ is the $p^{th}$ sample quantile, then the equality of theoretical and empirical qunatiles is expressed by 
$$Q(p_k;\theta)=\tilde{{Q}_p}_k$$
for $k=1, 2, ..., d$ with $d$, the number of parameters to be estimated. The MQE is available in the r package, “fitdistrplus” \cite{DelD15}. A numerical optimization is carried out to minimize the sum of squared differences between observed and theoretical quantiles. Thus, using the R-package, "fitdistrplus" the MQE of the SBKD can be obtained.
\section{Application to data}

\subsection{Simulation study}
It has been discussed under Subsection \ref{rngen} that a random sample of size $n$ can be generated from a SBKD using its quantile function. In this section some random samples with known parameters have been generated and the samples have been fitted to SBKD, Kumaraswamy distribution and Beta distribution respectively, by using the method of maximum likelihood estimation. The R package "fitdistrplus" has been used to obtain the MLE for the 3 distributions. The result obtained is summarized in Table \ref{simfit}.
\newpage
\begin{table}[ht]
\centering
\begin{adjustbox}{width=1\textwidth}
\small
\begin{tabular}{ccccccc}
\hline \hline
\textbf{SBKD($a$, $b$)}	& \textbf{Distribution}	& \textbf{Estimate 1}  & \textbf{Estimate 2}   & \textbf{Log-likelihood}  & \textbf{AIC}   & \textbf{BIC}\\
\hline
   		 		&SBKD	&0.8899073  &0.9473359  &18.75961   &-33.51922  &-28.30888 \\
$a=b=1$	 		&Kum	&1.903925   &0.945852   &18.75868   &-33.51735  &-28.30701 \\
		 		&Beta	&1.8919305  &0.9472815  &18.75947   &-33.51894  &-28.3086 \\
\hline
		 		&SBKD	&2.3282455  &0.7594830  &22.05229   &-40.10458  &-34.89424 \\
$a=2.3, b=0.75$	&Kum	&1.6561367  &0.7563391  &21.95013   &-39.90027  &-34.68993 \\
		 		&Beta	&1.6187877  &0.7597696  &22.01796   &-40.03593  &-34.82559 \\
\hline
		 		&SBKD	&3.055987   &2.220511   &43.08837   &-82.17674  &-76.9664  \\
$a=3, b=2$		&Kum	&3.896464   &2.278592   &43.07107   &-82.14215  &-76.93181 \\
				&Beta	&4.643706   &2.117711	&43.05      &-82.09999  &-76.88965 \\
\hline
 		 		&SBKD	&0.6176114  &1.5318908  &5.66006    &-7.320121  &-2.10978 \\
$a=0.65, b=1.6$	&Kum	&1.510307   &1.573296  	&5.652371   &-7.304742  &-2.094402 \\
		 		&Beta	&1.563505   &1.546740   &5.654174   &-7.308349  &-2.098008  \\
\hline \hline
\end{tabular}
\end{adjustbox}
\caption{MLE of the simulated datasets for SBKD, Kumaraswamy and Beta distributions}
\label{simfit}
\end{table}

Table \ref{simfit} clearly shows that, in case of simulated data from SBKD, the estimates are more closer to the actual values. The SBKD also gives a marginally better fit than the Kum and the Beta distribution in terms of the log-likelihood function. This is quite obvious as because the sample has been drawn from the SBKD. Figure \ref{se} gives a plot the standard error of the estimates, $\widehat{a}$ and $\widehat{b}$ of the simulated samples for increasing sample size. 

\begin{figure}[!htb]
\centering
\includegraphics[width=12cm]{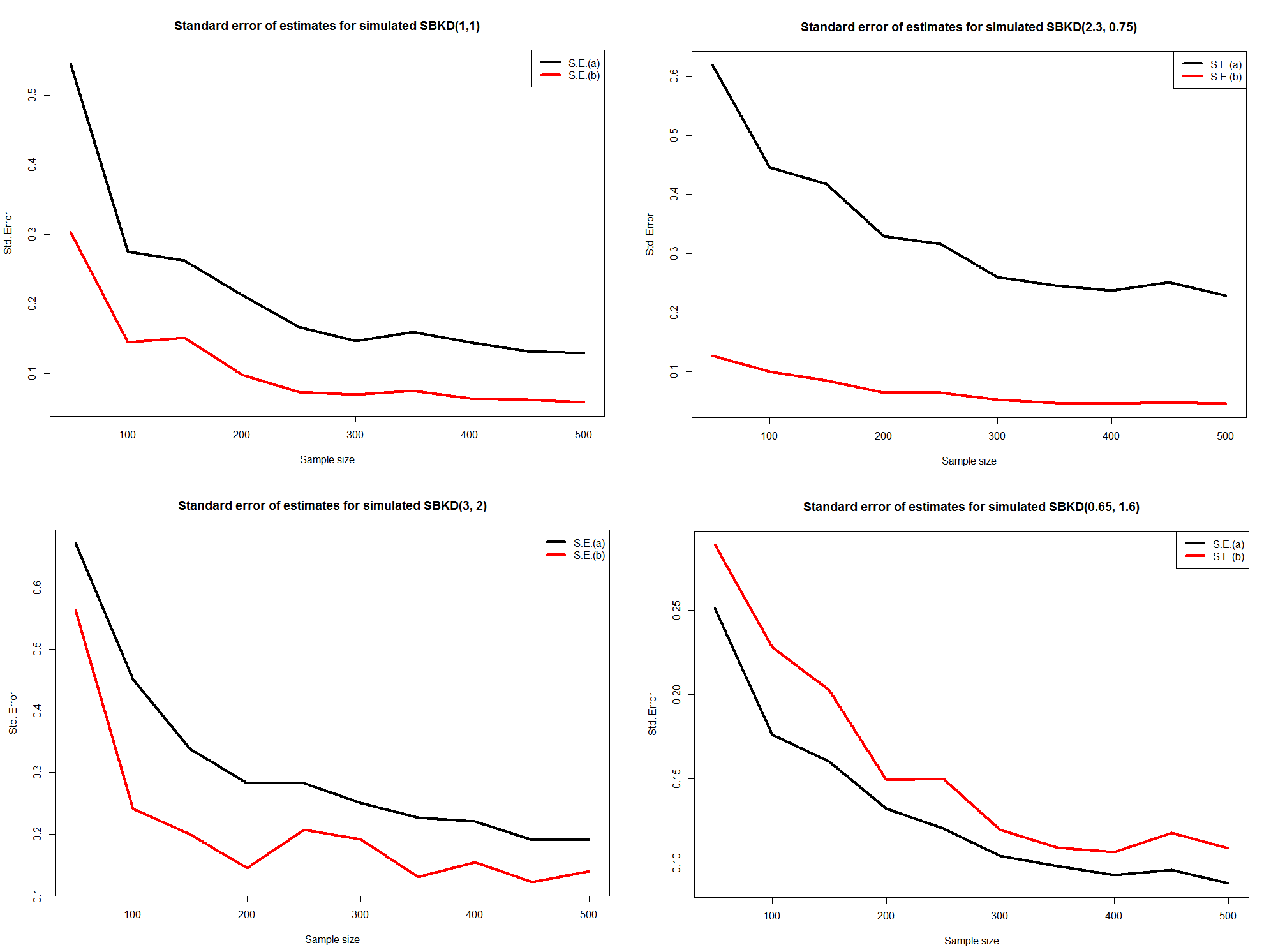}
\caption{Standard error of the estimates of $a$ and $b$ for an increasing sample size}
\label{se}
\end{figure}

Figure \ref{se} indicates that for all the simulated samples the standard error of the estimates decreases with increasing sample size. Hence the method of estimations as discussed in Section (\ref{Esti}) can be practically used to fit some real life data.

\subsection{Fitting to real life data}

In this section tensile strength data has been fitted to the size biased Kumaraswamy distribution by the method of MLE and MQE. The data is available in the R package \textit{gamlss.data} and it contains the measurements of tensile strength of 30 polyester fibres. R package \textit{fitdistrplus} has been used to obtain both the MLE and MQE. The above data fitted to the SBKD by the method of MLE and MQE is shown respectively in Figure \ref{tenmle} and \ref{tenmqe}.

\begin{figure}[!htb]
\centering
\includegraphics[width=10cm]{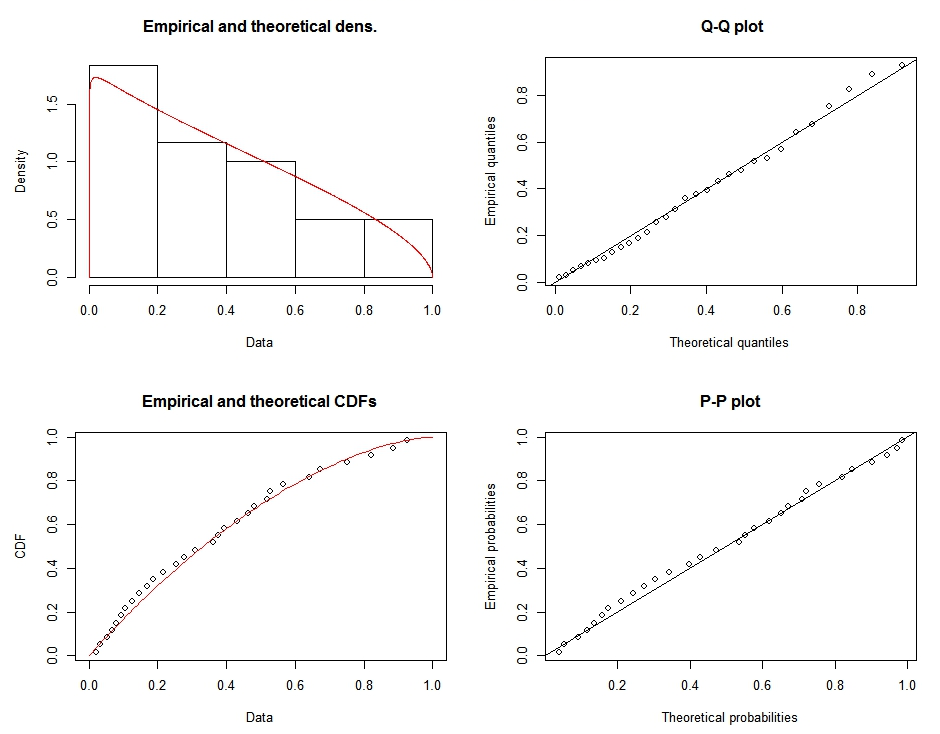}
\caption{Tensile strength data fitted to SBKD by the method of MLE}
\label{tenmle}
\end{figure}  

\newpage

\begin{figure}[!htb]
\centering
\includegraphics[width=10cm]{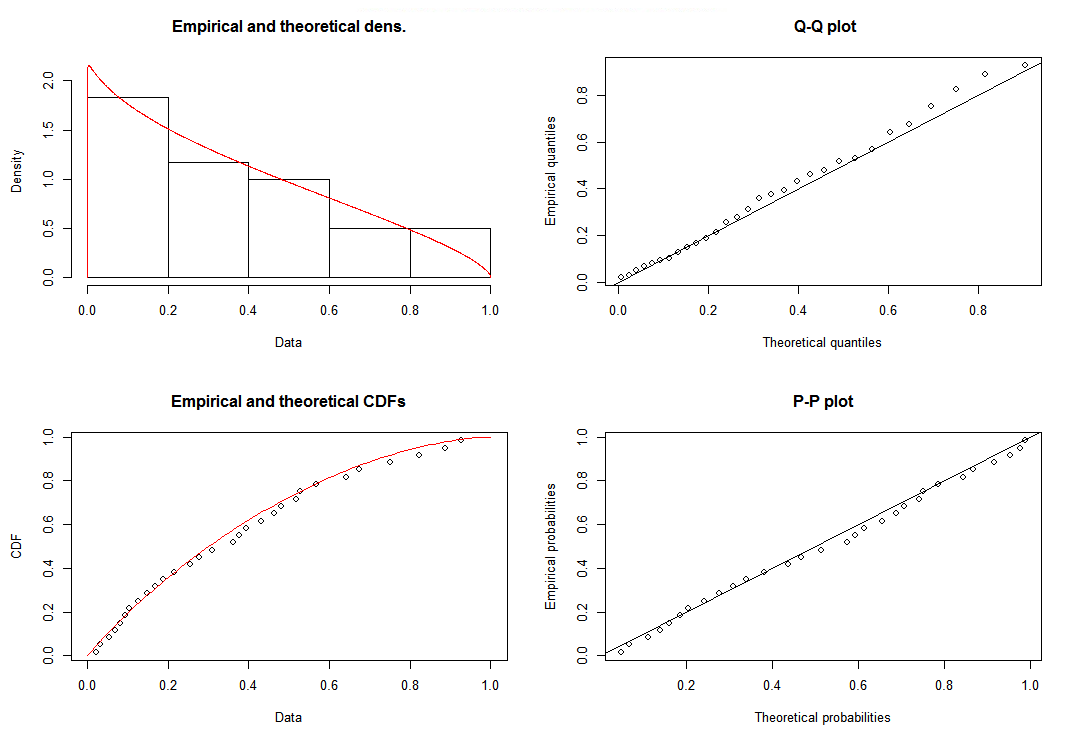}
\caption{Tensile strength data fitted to SBKD by the method of MQE}
\label{tenmqe}
\end{figure}  

The tensile data has also been fitted to the Kumaraswamy distribution and the beta distribution by the corresponding methods and the respective log-likelihood functions, the Akaike Information Criteria (AIC) and the Bayesian Information Criteria (BIC) have been obtained. The results obtained for the three distributions, viz., SBKD, Kum and Beta have been summarized in Table \ref{datatensile}.

\begin{table}[ht]
\centering
\begin{adjustbox}{width=1\textwidth}
\small
\begin{tabular}{ccccccc}
\hline \hline
\textbf{Method used}	& \textbf{Distribution}	& \textbf{Estimate 1}  & \textbf{Estimate 2}   & \textbf{Log-likelihood}  & \textbf{AIC}   & \textbf{BIC}\\
\hline
		 &SBKD	&0.1171583  &1.5864212  &3.422145  &-2.84429    &-0.04189501 \\
MLE		 &Kum	&0.9626828  &1.6082976  &3.311034  &-2.622069   &0.1803258 \\
		 &Beta	&0.9666515  &1.6204918  &3.305064  &-2.610127   &0.1922674 \\
\hline
		 &SBKD  &0.0895447  &1.6607512  &3.269999  &-2.539998   &0.2623964  \\
MQE      &Kum   &0.8160955  &1.4184062  &3.025052  &-2.050104   &0.7522909  \\
		 &Beta  &0.8026875  &1.4318374  &2.996381  &-1.992761   &0.8096336  \\
\hline \hline
\end{tabular}
\end{adjustbox}
\caption{Summary of the fitted datasets for SBKD, Kumaraswamy and Beta distributions}
\label{datatensile}
\end{table}
Table \ref{datatensile} clearly shows that in terms of the log-likelihood, the SBKD  gives a marginally better fit to the tensile strength data as compared to the Kum and Beta.

\section{Conclusions}
We have proposed size-biased version of Kumaraswamy distribution which can be employed in modeling data from hydrology, forestry and various other related fields. Special cases of the SBKD have been discussed. The structural properties including cumulative distribution function, the Quantile function, moments, and shape of the model for varying values of the parameters have been discussed and derived. Two methods for estimation of the parameters of the model viz,  MLE and MQE was studied. Using simulated data we have shown that the methods can provide reasonably good estimates of the parameters; it was shown that the standard deviations of the estimates decrease with increase in the sample size. The model has been applied to a real dataset which is indicative of potentially a better candidate than either a beta or a Kumaraswamy distribution in terms of greater likelihood.

\vspace{6pt}

\textbf{Acknowledgments: }The first author acknowledges the Department of Science and Technology (DST), Government of India for her financial support through DST-INSPIRE fellowship with award no. IF130343.


\begin{thebibliography}{}
\bibitem{DelD15}
Delignette-Muller, M.L.; Dutang, C. fitdistrplus: An R package for fitting distributions. {\em Journal of Statistical Software} {\bf 2015}, 64(4), 1-34.

\bibitem{DenP84}
Dennis, B.; Patil, G. The gamma distribution and weighted multimodal gamma distributions as models of population abtidance. {\em Mathematical Biosciences} {\bf 1984}, 68(2), 187-212.

\bibitem{DucG15}
Ducey, M.J.; Gove, H.G. Size-biased distributions in the generalized beta distribution family, with applications to forestry. {\em Forestry} {\bf 2015}, 88(1), 143-151.

\bibitem{Fis34}
Fisher, R.A. The effects of methods of ascertainment upon the estimation of frequencies. {\em Annals of Eugenics} {\bf 1934}, 6(1), 13-25.

\bibitem{00Gil}
Gilchrist, W.G. {\em Statistical modelling with quantile functions}; Chapman and Hall: New York, 2000; 

\bibitem{03Gov}
Gove, J.H. Estimation and applications of size-biased distributions in forestry. In {\em Modeling Forest Systems}; Amaro, A., Reed, D., Soares, P., Eds.; CABI Publishing: 2003; pp. 201-212.

\bibitem{Kum80}
Kumaraswamy, P. A Generalized probability density function for double-bounded random processes. {\em Journal of Hydrology} {\bf 1980}, 46(1), 79-88.

\bibitem{LapB87}
Lappi, J.; Bailey, R.L. Estimation of diameter increment function or other tree relations using angle-count samples. {\em Forest Science} {\bf 1987}, 33(3), 725-739.

\bibitem{Maget99}
Magnussen, S.; Eggermont, P.; Lariccia, V.N. Recovering tree heights from airborne laser scanner data. {\em Forest Science} {\bf 1999}, 45(3), 407-422.

\bibitem{McD84}
McDonald, J.B. Some generalized functions for the size distribution of income. {\em Econometrica} {\bf 1984}, 52(3), 647-663.

\bibitem{77PatR}
Patil, G. P.; Rao, C. R. Weighted distributions: a survey of their applications. In {\em Applications of Statistics};  Krishnaiah, P. R., Eds.; North Holland Publishing Company: 1977; pp. 383-405.

\bibitem{PatR78}
Patil, G.P.; Rao, C.R. Weighted distributions and size-biased sampling with applications to wildlife populations and human families. {\em Biometrics} {\bf 1978}, 34(2), 179-189.

\bibitem{65Rao}
Rao, C.R. On discrete distributions arising out of methods of ascertainment. In {\em Classical and Contagious Discrete Distributions}; Patil, G.P., Eds.; Pergamon Press and Statistical Publishing Society: Calcutta, 1965; pp. 320-332.

\bibitem{Taiet95}
Taillie, C.; Patil, G.P.; Hennemuth, R. Modeling and analysis of recruitment distributions. {\em Ecological and Environmental Statistics} {\bf 1995}, 2(4), 315-329.

\bibitem{09Tse}
Tse, Y. K.; Nonlife actuarial models: theory, methods and evaluation. In {\em International Series on Actuarial Science}; Cambridge University Press: 2009;

\bibitem{She72}
Scheaffer, R.L. Size-Biased Sampling. {\em Technometrics} {\bf 1972}, 14(3), 635-644.

\bibitem{Van86}
Van Deusen, P.C. Fitting assumed distributions to horizontal point sample diameters. {\em Forest Science} {\bf 1986}, 32(1), 146-148.

\bibitem{75War}
Warren, W.G. Statistical distributions in forestry and forest products research. In {\em Statistical Distributions in Scientific Work Vol. 2.}; Patil, G.P., Kotz, S., Ord, J.K., Eds.; D. Reidel: 1975; pp. 369-384.

\bibitem{Yuet12}
Ye, Y.; Oluyede, B.O.; Pararai, M. Weighted generalized beta distribution of the second kind and related distributions. {\em Journal of Statistical and Econometric Methods} {\bf 2012} 1(1), 13-31.
\end{thebibliography}
\end{document}